\documentclass[conference,letterpaper]{IEEEtran}

\usepackage{amsthm}
\usepackage{amsmath} %
\usepackage{mathtools}
\usepackage{graphicx}
\usepackage{psfrag}
\usepackage{epsf,epsfig}
\usepackage{epstopdf}
\usepackage{subfigure}
\usepackage{cite,color}
\usepackage{cite}
\usepackage{amsmath,amssymb,amsfonts}
\usepackage{algorithmic}
\usepackage{graphicx}
\usepackage{textcomp}
\usepackage{xcolor}
\usepackage{verbatim}
\usepackage{amsmath}

\usepackage{multirow}

\usepackage{tikz}

\definecolor{red}{rgb}{1,0,0}  

\newcommand{\CNcal}{\mathcal{CN}}

\newcommand{\snr }{{\rm snr }}

\usepackage{thmbox}

\newtheorem{Cor}{Corollary}

\newtheorem{Prop}{Proposition}

\begin{document}
\title{Closed-Loop Binary Media-Based Modulation \vspace{-.05cm}}

\author{

\vspace{-.000cm}

\IEEEauthorblockN{ Majid Nasiri Khormuji and Branislav M. Popovic} \\
\vspace{-.5cm}
\IEEEauthorblockA{\textit{Huawei Technologies Sweden AB, Stockholm, Sweden}\\
                           \{majid.nk, branislav.popovic\}@huawei.com
                          \vspace{-.3cm}
                           }
}

\maketitle

\begin{abstract}

Presenting analytical results for Binary Media-Based Modulation (B-MBM) over fading channels for single-antenna receivers. Illustrating that open-loop B-MBM, in the absence of feedback, only achieves a diversity order of one. However, with feedback and optimal weight selection in closed-loop configurations, a diversity order of two becomes achievable. Notably, the closed-loop B-MBM, with analytically computed optimal weights, performs equivalent to Alamouti-coded BPSK transmission, demonstrating feasibility even with just one radio frequency chain when feedback is available.

\end{abstract}

\begin{IEEEkeywords}
Binary Media-Based Modulation, Feedback, Transmit Diversity, and RF Mirrors.
\end{IEEEkeywords}

\section{Background and Outline}
The usage of media to enhance information transmission and reception has garnered significant attention, where an information-carrying signal on its way to an indented
destination is modified to enable a more reliable wireless link.  To achieve this, the components known as
\lq mirrors\rq~are employed, facilitating  advantageous signal  modification. Specifically, the deployment of Radio Frequency (RF) mirrors in close proximity of the transmit antennas or encircling them,  has enabled the technique known as Media-Based Modulation (MBM) \cite{MBM_khandani_ISIT, MBM_khandani_ISIT2, MBM_Seifi, Basar}.
Through utilization of the RF mirrors, the signal emitted by the transmitter undergoes \emph{shaping} prior reaching its  destination. Notably, the MBM holds significant potential for new use cases, with prospects for its integration into upcoming 6G networks\cite{Basar}.



This paper considers Binary MBM, focusing specifically on scenarios where a  two-state mirror is positioned  in proximity to a transmit antenna. This is the lowest order MBM variant serving as an introductory model, aiding in establishment of a fundamental understanding of MBM,  especially when incorporating receiver feedback into the transmission process.   This transmission mode, labeled as \emph{closed-loop binary MBM}, is systematically examined.

The investigation  begins by computing the Bit Error Rate (BER) of the open-loop B-MBM configuration over Rayleigh fading channels. The findings reveal that the open-loop approach fails to yield any appreciable diversity gain when compared to the conventional single-antenna setup employing BPSK transmission. To address this limitation,  a closed-loop B-MBM configuration is set up. The closed-loop B-MBM employs  multiplicative complex weights derived from the channel coefficients. The optimal weights are analytically found by which it is proved that the closed-loop B-MBM achieves  performance on par with BPSK transmission  using the celebrated Alamouti space-time code in \cite{Alamouti} for two transmit antennas. However, a crucial aspect highlighted is that B-MBM effectively utilizes only one RF chain.


The remainder of the paper is organized as follows. Section~\ref{sec:Open_Binary MBM} outlines the baseline B-MBM  scheme and provides its analytical BER expression.    Section~\ref{sec:Close_Binary MBM} presents the closed-loop B-MBM with analytically obtained optimal complex weights.  Section~\ref{sec:performance_CL_BMBM} computes the BER of the optimal closed-loop B-MBM in Section~\ref{sec:Close_Binary MBM}. Section~\ref{sec:Alamouti} discusses the connection of the optimal closed-loop B-MBM to Alamouti-coded BPSK transmission. Section~\ref{sec:Analog_Close_Binary MBM} presents the closed-loop B-MBM with unit-amplitude weights. Section~\ref{sec:Peformance_simulation} presents simulations results and compares them with analytical derivations. Section~\ref{sec:concl} summarizes the major results.



%


\section{Open-Loop Binary MBM}\label{sec:Open_Binary MBM}

\begin{figure}[t]
\centering
\vspace{-0.0cm}
	\includegraphics[width=.45\textwidth]{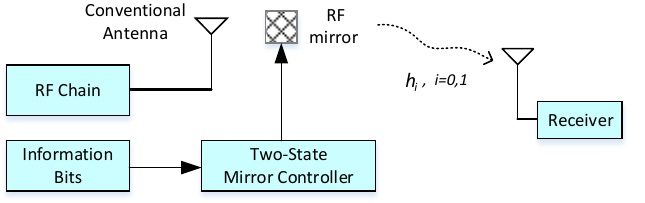}  
\vspace{-0.2cm}
	\caption{B-MBM transmission link with a single RF chain, a two-state RF connected to a controller and a single-antenna receiver.}
	\label{Fig:signaling}
	\vspace{-0.4cm}
\end{figure}

Fig. 1 depicts the basic block diagram of the transmitter utilizing binary media-based modulation. A \emph{single }antenna, which is connected to an RF chain, emits a radio wave at a specified frequency. The transmit antenna is (partly) surrounded by a mirror (a.k.a. a reflective surface) which is suitably positioned to act as a channel scatterer to enable desirable channel properties. Without loss of generality, the diagram displays  the mirror located in the proximity of the transmit antenna. The emitted signal passes through this mirror before reaching its intended destination.  The binary information bits are passed to a mirror controller where it based on the input
bit generates a signal to configure  the mirror's pattern in either of its two distinct states. In other words, the transmitter employs two states to create \emph{two }channel realizations which are mapped to a single bit of information (i.e. either $b=0$ or $1$).  We refer to this configuration as open-loop \emph{Binary MBM (B-MBM)} as there is no feedback involved in the transmission process.

The constellation signal points of B-MBM are therefore
\begin{eqnarray}\label{eq:BMBM}
    &&	s_0(b=0) = h_0= g_0 e^{j\theta_0} \\ 
     &&   s_1(b=1) = h_1= g_1 e^{j\theta_1},
\end{eqnarray}
where $h_i$ denotes the channel realization for each state of the mirror,  $s_i$ denotes the received signal point the receiver, $g_i$ is the amplitude of the channel and $\theta_i$ is the corresponding phase of the channel coefficient, for $i=0,1$. We assume throughout that the channels are independent unit-variance Rayleigh fading variables, i.e. $h_i \sim \CNcal (0,1)$.

The received discrete-time baseband signal at the destination, $y$, is therefore given by
\begin{eqnarray}\label{eq:Rx_signal}
    	 y= s(b)+ z,
\end{eqnarray}
where $z$ denotes zero-mean additive Gaussian noise  with variance of $N_0$.
This leads to a B-MBM scheme with a signal constellation whose points  are separated by
\begin{eqnarray}\label{eq:BMBM_distance1}
    	 \Delta= s_1(b=1) - s_0(b=0) =  g_1 e^{j\theta_1} - g_0 e^{j\theta_0}.
\end{eqnarray}
Thus, based on \eqref{eq:BMBM_distance1}, and the assumption on the channel statistics, it can be shown that
\begin{eqnarray}\label{eq:BMBM_distance}
    	 \Delta  \sim \CNcal (0,2).
\end{eqnarray}
Therefore, the open-loop B-MBM gives rise to the distance of
\begin{eqnarray}\label{eq:BMBM_distance}
    	d = | \Delta |
    \sim \mathrm{Rayleigh}(0,\sqrt{2}).
\end{eqnarray}
That is, the distribution of distance between the constellation points is zero-mean Rayleigh. This is hence not a good signal constellation on the average. The following proposition shows that open-loop B-MBM only achieves a diversity order of one.\footnote{ The diversity order is defined as the exponent of the decay of error at the receiver as a function of SNR \cite{Tse05}.}

\begin{Prop}
    Bit Error Rate (BER) of the open-loop B-MBM is
    \begin{eqnarray}
    	P_b^{\text{(OL B-MBM)}}  &=& \frac{1}{2} \left( 1- \sqrt{\frac{\snr}{ 2+ \snr}} \right) \nonumber \\
       &=& \frac{1}{2 \snr } + O\left(\frac{1}{ \snr^2 }\right),
    \end{eqnarray}
    where $\snr:=N_0^{-1}$ for unit power transmitter,  unit-variance Rayleigh fading and AWGN with variance of $N_0$ at the receiver. Hereafter, $f(x) = O (g(x))$ means that there exists $\Omega,M \in R$
    such that $\left|\frac{f(x)}{g(x)}\right| \leq M$ whenever $x> \Omega$.
\end{Prop}

\begin{proof}
The open-loop B-MBM generates the received signal constellation with the distance of $|h_1-h_0|$.
 Let $h_{\mathrm{eq}}:=h_1-h_0$ and assume that we transmit BPSK symbols with half the amplitude (i.e. $s_i = \pm \frac{1}{2}$)  over $h_{\mathrm{eq}}$. This setup creates $y=h_{\mathrm{eq}}s_i +z $, whose received signal constellation has the distance of $h_{\mathrm{eq}}$. Since $h_{\mathrm{eq}}$ still has a Rayleigh distribution with twice variance, we can use BER of BPSK. The BER of BPSK over Rayleigh fading is given by \cite{Ahlin}
\begin{eqnarray}
    	P_b^{(\text{BPSK})}=\frac{1}{2} \left( 1- \sqrt{\frac{\snr}{ 1+ \snr}} \right).
\end{eqnarray}
Using the equivalent SNR for the open-loop B-MBM, the corresponding BER is hence
  \begin{eqnarray}
    	P_b^{\text{(OL B-MBM)}}  = \frac{1}{2} \!\left( 1- \sqrt{\frac{\snr/2}{ 1+ \snr/2}} \right) = \frac{1}{2 \snr } + O\!\left(\!\frac{1}{ \snr^2 }\!\right).
    \end{eqnarray}
The last equality follows by appropriate Taylor expansions, thereby concluding the proof.
\end{proof}

\begin{figure}
\centering
\vspace{-0.0cm}
	\includegraphics[width=.47\textwidth]{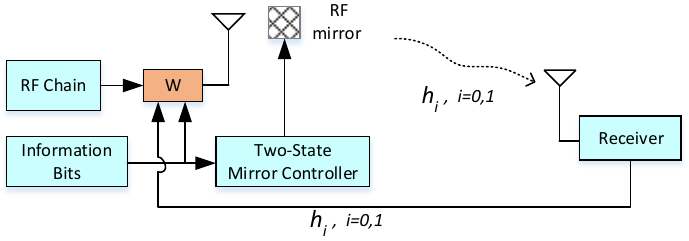} 
\vspace{-0.2cm}
	\caption{Closed-loop B-MBM transmission link with a single RF chain,  a two-state mirror and complex weight  $w$ adjusted by feedback.}
	\label{Fig:signaling}
	\vspace{0.1cm}
\end{figure}
\section{Closed-Loop Binary MBM}\label{sec:Close_Binary MBM}
In the closed-loop B-MBM, the signal prior reaching each RF mirror  is multiplied by a complex coefficient that depends on the information bits and the channel coefficients. That is
\begin{eqnarray}\label{eq:BMBM1}
    	s_0(b=0) &=& w\left(b=0,h_0,h_1\right) \cdot h_0 \\
        s_1(b=1) &=& w\left(b=1,h_0,h_1\right)  \cdot h_1,\label{eq:BMBM2}
\end{eqnarray}
where $w(b,h_0,h_1)$ denotes the weights for $b=0,1$.  The weights are subject to the power constraint at the transmitter. To compute the average weight power using two above arbitrary weights, consider
\begin{align}\label{eq:Avg_Weight}
    	\mathbb{E} \left[|w|^2\right] &= \int \left|w(b,h_0,h_1)\right|^2 p(b,h_0,h_1)db dh_0 dh_1,
\end{align}
where $p(b,h_0,h_1)$ is the joint probability of input bits $b$, and channel coefficients $h_0,h_1$.
Since $b$ is a discrete random variable having two possible values (0 and 1), while $h_0$ and $h_1$ are two continuous random variables,
$\eqref{eq:Avg_Weight}$ simplifies to
\begin{align}\label{eq:Avg_Weight2}
    &	\mathbb{E} \left[|w|^2\right]\nonumber\\
    &= p_b (b=0) \int \left|w(b=0,h_0,h_1)\right|^2 p(h_0,h_1|b=0) dh_0 dh_1\nonumber\\
     &+ p_b (b=1) \int\left|w(b=1,h_0,h_1)\right|^2 p(h_0,h_1|b=1)dh_0 dh_1 \nonumber\\
    &=  p_b (b=0) \int \left|w_0(h_0,h_1)\right|^2 p(h_0,h_1|b=0) dh_0 dh_1\nonumber\\
                                     &+ p_b (b=1) \int\left|w_1(h_0,h_1)\right|^2 p(h_0,h_1|b=1)dh_0 dh_1,
\end{align}
where the last equality follows by setting $w(b=0,h_0,h_1 )=w_0 (h_0,h_1)$  and $w(b=1,h_0,h_1 )=w_1 (h_0,h_1)$.
As the channel realizations are independent of the information bits, we therefore have
\begin{eqnarray}\label{eq:Joint_pdf_channel}
    p(h_0,h_1|b=0)=p(h_0,h_1|b=1)=p(h_0,h_1).
\end{eqnarray}
By further assuming that $p_b (b=0)=p_b (b=1)=0.5$, from  \eqref{eq:Avg_Weight2} and \eqref{eq:Joint_pdf_channel}, we obtain
\begin{align}\label{eq:Avg_Weight3}
      &	\mathbb{E} \left[|w|^2\right]=\nonumber\\
      & \frac{1}{2} \int \left[ \left|w_0(h_0,h_1)\right|^2 + \left|w_1(h_0,h_1)\right|^2 \right] p(h_0,h_1) dh_0 dh_1.
\end{align}
Therefore, to obtain the unit average power constraint on the weights at the transmitter where $\mathbb{E} \left[|w|^2\right]\leq 1$, it is sufficient to meet the condition
\begin{align}\label{eq:power_constraint}
       \left|w_0(h_0,h_1)\right|^2 + \left|w_1(h_0,h_1)\right|^2 \leq 2.
\end{align}
In the sequel, for ease of presentation we set  $w_i:=w_i(h_0,h_1)$ for $i=0,1$.
Based on the signal points in \eqref{eq:BMBM1} and \eqref{eq:BMBM2}, the maximum likelihood  detection rule at the receiver then becomes the following.  For the received signal $y=s_i + z$, the receiver declares the transmitted bit as $0$ if
\begin{eqnarray}\label{eq:BMBM}
    |y - w_0h_0|^2 < |y - w_1h_1|^2,
\end{eqnarray}
otherwise, it declares 1 as the transmitted bit. Therefore, it is optimal to maximize the distance between the constellation points.  That is, the optimal weights \footnote{Throughout,  $(\cdot)^{\ast}$  denotes the corresponding optimal value and $(\cdot)^{\dagger}$  denotes the complex conjugate operation.} can be found based on
\begin{eqnarray}\label{eq:BMBM}
   && 	w_i^{\ast} = \ \max \arg_{w_i} \left|w_0h_0 - w_1 h_1\right | \\
   && \ \ \ \  \ \ \ \ \text{s.t.}  \ \ |w_0|^2 +  |w_1|^2 \leq 2.
\end{eqnarray}

The following proposition provides closed-form expressions of the weights when perfect knowledge of both instantaneous channel coefficients are available at the transmitter.
\begin{Prop}\label{prop:opt_weight}
The optimal weights for the closed-loop B-MBM are given by
 \begin{eqnarray}\label{eq:BMBM_optimalweights}
   && 	w_0^{\ast} =\frac{\sqrt{2}}{\sqrt{|h_0|^2+|h_1|^2}} |h_0| \\
   &&   w_1^{\ast} =\frac{\sqrt{2} e^{j\pi}}{\sqrt{|h_0|^2+|h_1|^2}} \frac{h_0 h_1^{\dagger}}{|h_0|}.
\end{eqnarray}

\end{Prop}

\begin{proof}
Let
\begin{eqnarray}
   w_i = a_i e^{j\phi_i}, \  \text{and} \  h_i = |h_i| e^{j\theta_i}.
\end{eqnarray}
The received constellation points are depicted in Fig.~3. The distance between the points can be  computed according to
\begin{eqnarray}\label{angle}
  d^2&=& a_0^2 |h_0|^2 + a_1^2 |h_1| ^2 \nonumber\\
  &&-  2a_0a_1|h_0| |h_1| \cos(\theta_1+\phi_1 - \theta_0-\phi_0 ).
\end{eqnarray}
Therefore, we note that for any choice of $a_i$, the distance between the two points is maximized if  $\theta_1+\phi_1 - \theta_0-\phi_0=\pi$.  That is, the line connecting the points passes through the origin.
Without loss of generality, we can set
\begin{eqnarray}\label{eq:BMBM}
   && \phi_0 =  0 \\
   && \phi_1 = \pi + \theta_0 - \theta_1.
\end{eqnarray}
This rotates the constellation point $s_1$ and changes it to $s_1^{\dagger}$ as illustrated in Fig.~3.

This simplifies the optimization problem to
\begin{eqnarray}\label{eq:BMBM_obejective}
   && 	\max  \ \  a_0|h_0| + a_1 |h_1 | \label{eq:objective} \\
   && \text{s.t.} \ a_0^2 +  a_1^2\leq 2.
\end{eqnarray}
We observe that the optimal values $(a_0^{\ast}$,   $a_1^{\ast})$ should satisfy the constraint $a_0^2 +  a_1^2 \leq 2$ with equality.
This can be proved by contradiction. Assume that the optimal pair  $(a_0^{\ast},a_1^{\ast})$
 yields  $(a_0^{\ast})^2 +  (a_1^{\ast})^2 < 2$.
 Now let $a_0=a_0^{\ast}+ t$ such that $t$ is chosen to
satisfy $(a_0^{\ast}+t)^2 +  (a_0^{\ast})^2 =2$.  For this new pair, the objective function $(a_0^{\ast}+ t)|h_0| + a_1^{\ast} |h_1 |$ becomes larger  than \eqref{eq:objective}, since $t>0$. Therefore, the optimal solution should satisfy the power constraint by equality.   Thus, we can assume that
\begin{eqnarray}\label{eq:equality_constraint}
  a_1= \sqrt{2 - a_0^2}.
\end{eqnarray}
Replacing back \eqref{eq:equality_constraint} into \eqref{eq:BMBM_obejective}, we thus obtain
\begin{eqnarray}
    	\max a_0 |h_0| + (2-a_0^2)^{\frac{1}{2}} |h_1| =: \max f(a_0),
\end{eqnarray}
where $0<a_0<\sqrt{2}$.
To find the optimal solution, we solve
\begin{eqnarray}
    	\frac{\partial f}{\partial a_0}=0.
\end{eqnarray}
This yields
\begin{eqnarray}\label{optimal_value}
    	a_0^{\ast} =\sqrt{2}\frac{|h_0|}{\sqrt{|h_0|^2+|h_1|^2}},
\end{eqnarray}
and
\begin{eqnarray}
    	a_1^{\ast} =\sqrt{2}\frac{|h_1|}{\sqrt{|h_0|^2+|h_1|^2}}.
\end{eqnarray}

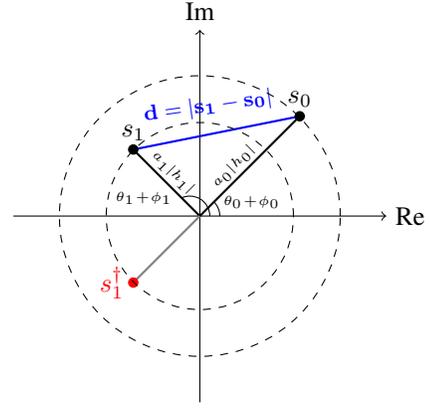
\begin{figure}[t]\label{Fig:opt_const}
\centering
\begin{tikzpicture}[scale=0.885]
\draw[->] (-2.8,0) -- (2.8,0) node[right] {Re};
\draw[->] (0,-2.8) -- (0,2.8) node[above] {Im};
\draw[blue, thick](1.5,1.5) -- (-1,1);
\node[blue,rotate=8] (d) at (.125,1.65) {\footnotesize$\mathbf{d=|s_1-s_0|}$};
\draw[ thick] (0,0) -- (-1,1);
\draw[ thick] (0,0) -- (1.5,1.5);
\draw[gray,thick] (0,0) -- (-1,-1);
\filldraw[black] (1.5,1.5) circle (2pt) node[above] {$s_0$};
\filldraw[black] (-1,1) circle (2pt)node[above] {$s_1$};
\filldraw[red] (-1,-1) circle (2pt)node[left] {$s_1^{\dagger}$};
\draw [dashed, line width=.4pt] (0,0) circle (40pt);
\draw [dashed, line width=.4pt] (0,0) circle (60pt);
\draw (3mm,0mm) arc (0:45:3mm)node[right] {\tiny $\theta_0\!+\! \phi_0$};
\draw (1.5mm,0mm) arc (0:115:3mm)node[left] {\tiny $\theta_1\!+\! \phi_1$};;
\node[rotate=45] at (.51,.751) {\tiny $a_0|h_0|$};
\node[rotate=315] at (-.4,.651) {\tiny $a_1|h_1|$};
\end{tikzpicture}
\vspace{-0.2cm}
\caption{Closed-loop B-MBM received constellation points with rotation.}
\vspace{-0.2cm}
\end{figure}

To prove that  \eqref{optimal_value} produces the maximum of $f(a_0)$, we compute the second order derivative and show that its value is negative. Consider
\begin{eqnarray}
    	\frac{\partial^2 f}{\partial a_0^2}&=& - (2-a_0^2)^{-\frac{1}{2}} |h_1| -  a_0^2 (2-a_0^2)^{-\frac{3}{2}} |h_1|  \nonumber\\
    &=&-|h_1|(2-a_0^2)^{-\frac{1}{2}} {\frac{2 }{2-a_0^2}}.
\end{eqnarray}
Since at the chosen points we have  $ 2 -(a^{\ast}_0)^2 = (a^{\ast}_1)^2 > 0$,
\begin{eqnarray}
    	\frac{\partial^2 f}{\partial a_0^2} (a_0=a_0^{\ast})<0,
\end{eqnarray}
which proves the optimality and completes the proof.

\end{proof}


\section{Performance of Closed-Loop  B-MBM}\label{sec:performance_CL_BMBM}
In this section, we analytically compute the BER of closed-loop B-MBM configuration  employing the optimal weight selection outlined in Prop.~2. The analysis indicates that the BER decays proportional  to $\frac{1}{\snr^2 } $ for large values of $\snr$. This therefore  shows that the closed-loop B-MBM attains the diversity order of two, highlighting a substantial performance leap  when contrasted  to  the open-loop B-MBM configuration when feedback is accessible at the transmit side.

\vspace{1cm}
%
%
%
%
%
%
%



{}
\vspace*{.19cm}
\begin{Prop}
    BER of the closed-loop B-MBM is
 \vspace{-.22cm}
\begin{eqnarray} \label{eq:BER_B_MBM}
    \!\!\!\!\!\!&&P_b^{\text{(CL B-MBM)}} \nonumber\\
 \!\!\!\!\!\! &=&\!\!\frac{1}{4} \left( 1- \sqrt{\frac{\snr}{ 2+ \snr}} \right)^2\!\! \left[ 1 + \frac{1}{2} \left(1 + \sqrt{\frac{\snr}{2 + \snr}} \right)^2\right]\nonumber\\
 \!\!\!\!\!\! &=&\!\!\frac{3}{4 \snr^2 } +  O\left(\frac{1}{ \snr^3 }\right).
\end{eqnarray}
\vspace{-.22cm}
\end{Prop}

\begin{proof}
The optimal B-MBM generates the received constellation with the distance
\begin{eqnarray}\label{eq:BMBM_dist}
    	d^{\ast}&=& a_0^{\ast} |h_0| + a_1^{\ast} |h_1|
           = \sqrt{2} \sqrt{|h_0|^2+|h_1|^2}.
\end{eqnarray}

We recall that the Maximum Ratio Combing (MRC) precoding of a BPSK symbol transmitted simultaneously with two transmit RF chains is a transmit diversity scheme with weights
$w_i ={h_i^{\ast}}/{\sqrt{|h_0|^2+|h_1|^2}}$ for $i=0,1$, which makes the total unit transmit power since $|w_0|^2+|w_1 |^2=1$.
The received signal $y$ is the sum of the transmitted signals multiplied by respective channel coefficients, which thus can be represented as
\begin{eqnarray}
    	y  &=& \sqrt{|h_0|^2+|h_1|^2} \  x + z,
\end{eqnarray}
at the receiver.  The received BPSK signal constellation points then become
\begin{eqnarray}
    	s_0  = +\sqrt{{|h_0|^2+|h_1|^2}}, \
        s_1  = - \sqrt{{|h_0|^2+|h_1|^2}}.
\end{eqnarray}
This leads to the distance between the received constellation points which is equal to
  \begin{eqnarray}\label{eq:distance_relation_mrc}
    	d_{\texttt{MRC}}  &=& 2\sqrt{|h_0|^2+|h_1|^2}=\sqrt{2}d^{\ast}.
\end{eqnarray}
Given the channel statistics are identical in both scenarios, we can evaluate the  BER of the closed-loop B-MBM via BER of the MRC precoded BPSK, given in \cite{Ahlin} as
\begin{eqnarray}\label{eq:BER_MRC}
    	P_b^{\text{(MRC)}} \! =\! \frac{1}{4}\! \left( 1- \sqrt{\tfrac{\snr_{\text{MRC}}}{ 1+ \snr_{\text{MRC}}}} \right)^2 \!\! \left[ 1 + \frac{1}{2} \left(1 \! + \! \sqrt{\tfrac{\snr_{\text{MRC}}}{1  + \snr_{\text{MRC}}}} \right)^2\right],
\end{eqnarray}
by replacing $\snr_{\text{MRC}}$ in \eqref{eq:BER_MRC} by $\tfrac{1}{2}\snr$, because of \eqref{eq:distance_relation_mrc}.
\end{proof}

%
%
%
%
%



\section{Connection to Alamouti Code}\label{sec:Alamouti}

\begin{Prop}
The optimal closed-loop B-MBM performs as good as Alamouti-coded BPSK scheme.
\end{Prop}

\begin{proof}

Recall that the Alamouti-coded \cite{Alamouti} received signal from two antennas after post-processing at the receiver is given by \footnote{We have added an additional power  normalization factor of $\sqrt{2}$ as compared to Eq. (13) in \cite{Alamouti} to keep the total transmit power unchanged. }
\begin{eqnarray}
    	y  &=& \sqrt{\tfrac{|h_0|^2+|h_1|^2}{2}} x + z.
\end{eqnarray}
That is, when the unit-power BPSK signal is transmitted, the received signal constellation points becomes

\begin{eqnarray}
    	s_0 =+\sqrt{\tfrac{|h_0|^2+|h_1|^2}{2}}, \         s_1  = - \sqrt{\tfrac{|h_0|^2+|h_1|^2}{2}}.
\end{eqnarray}
This leads to the distance between the signal points which is
  \begin{eqnarray}
    	d  &=& \sqrt{2(|h_0|^2+|h_1|^2)}.
\end{eqnarray}
Establishing that $d$ is equal to $d^{\ast}$ in \eqref{eq:BMBM_dist} completes  the proof.
\end{proof}


\begin{table*}\label{Table_summary}
  \centering
  \vspace{-.002cm}
  \caption{Comparison of BPSK with Open-loop (OL) and Closed-Loop (CL) B-MBM.}
  \vspace{-.002cm}
\begin{tabular}{| c | c |  c | c | c| c| c|}
\hline
 \multirow{2}*{ Mod.} & \multirow{2}*{Scheme}& \# of RF   &  Feed- & \multirow{2}*{Bit Error Rate (BER)}& High-SNR & Div.  \\
& &  Chains & Back  & & Approximation & Order  \\
 \hline  \hline
  & SISO & 1 & No & $P_b^{\text{(SISO)}}  = \frac{1}{2} \left( 1- \sqrt{\frac{\snr}{ 1+ \snr}} \right)$  & $ \frac{1}{4 \snr } + O\left(\frac{1}{ \snr^2 }\right)$  & 1 \\
 \cline{2-7}
 \multirow{1}{*}{ \textbf{BPSK}}   & MRC & 2 & Yes& $P_b^{\text{(MRC)}}  = \frac{1}{4} \left( 1- \sqrt{\frac{\snr}{ 1+ \snr}} \right)^2 \left[ 1 + \frac{1}{2} \left(1 + \sqrt{\frac{\snr}{1 + \snr}} \right)^2\right] $& $ \frac{3}{16 \snr^2 } +  O\left(\frac{1}{ \snr^3 }\right) $ & 2\\
 \cline{2-7}
    & Alamouti & 2 & No & $P_b^{\text{(Alamouti)}}
 =\frac{1}{4} \left( 1- \sqrt{\frac{\snr}{ 2+ \snr}} \right)^2 \left[ 1 + \frac{1}{2} \left(1 + \sqrt{\frac{\snr}{2 + \snr}} \right)^2\right]$ &$ \frac{3}{4 \snr^2 } +  O\left(\frac{1}{ \snr^3 }\right) $ & 2 \\
  \hline   \hline
  & OL: $w_i\!=\!1$& 1 &No  & $P_b^{\text{(OL B-MBM)}}  = \frac{1}{2} \left( 1- \sqrt{\frac{\snr}{ 2+ \snr}} \right)$& $ \frac{1}{2 \snr } + O\left(\frac{1}{ \snr^2 }\right)$ & 1  \\
  \cline{2-7}
  \multirow{1}{*}{\textbf{B-MBM}}  & CL: $w_i\!=\!a_i e^{j\phi_i}$ & 1 & Yes & $P_b^{\text{(CL B-MBM)}} =
    \frac{1}{4} \left( 1- \sqrt{\frac{\snr}{ 2+ \snr}} \right)^2 \left[ 1 + \frac{1}{2} \left(1 + \sqrt{\frac{\snr}{2 + \snr}} \right)^2\right]$ &$ \frac{3}{4 \snr^2 } +  O\left(\frac{1}{ \snr^3 }\right) $ & 2\\
     \cline{2-7}
    & CL: $w_i\!=\!e^{j\phi_i}$ & 1 & Yes & $P_{b, \text{Unit-Amp}}^{\text{(CL B-MBM)}}  =\frac{1}{2} \left( 1-  \frac{\sqrt{\snr (\snr+4)}}{\snr+2} \right)$ &$ \frac{1}{ \snr^2 } +  O\left(\frac{1}{ \snr^3 }\right) $ & 2\\
  \hline
\end{tabular}
\end{table*}


\section{Closed-Loop B-MBM with Unit-\\Amplitude Weights}\label{sec:Analog_Close_Binary MBM}
In some practical applications, optimizing only the signal phase can prove beneficial due to the energy efficiency characteristics of  power amplifiers, allowing the amplitude of the transmit signal to remain unchanged.
This approach not only facilitates an \emph{analog }implementation of the closed-loop B-MBM utilizing phase shifters exclusively but also allows for integration with the RF mirrors, effectively \emph{decoupling} it from the conventional antenna during the transmission process (i.e. the W-box is placed on the RF mirror in Fig. 2).

By the results in Section~\ref{sec:Close_Binary MBM},  the optimized \emph{unit-amplitude} weights when perfect knowledge of both instantaneous channel coefficients are available at the transmitter, are given by
 \begin{align}\label{eq:BMBM_optimal_analog_weights}
  	w_0^{\ast} &= 1  ,  \ \ \    w_1^{\ast} =\frac{h_0 h_1^{\dagger}}{|h_0| |h_1|}e^{j\pi}.
\end{align}
This yields the following BER result.
\begin{Prop}
    BER of  the closed-loop B-MBM with unit-amplitude weights is
    \begin{align}\label{eq:BER_analog_B_MBM}
    P_{b, \text{Unit-Amp}}^{\text{(CL B-MBM)}}  &=\frac{1}{2} \left( 1-  \frac{\sqrt{\snr (\snr+4)}}{\snr+2} \right)\nonumber\\
                                              &=\frac{1}{\snr^2 } +  O\left(\frac{1}{ \snr^3 }\right).
\end{align}
\end{Prop}

\begin{proof}
The optimal closed-loop unit-amplitude B-MBM using the above  weights in \eqref{eq:BMBM_optimal_analog_weights} generates the received signal constellation points with the distance
\begin{align}\label{eq:Analog_BMBM_dist}
    	d^{\ast}_{\text{Unit-Amp}}=\left|s_0 -s_1\right| &= \left|h_0 -h_1\frac{h_0 h_1^{{\dagger}}}{|h_0||h_1|}e^{j\pi}\right| \nonumber \\
                 &= \left|h_0 +\frac{h_0 |h_1|}{|h_0|}\right| \nonumber \\
                  &= |h_0|\left| 1 +\frac{|h_1|}{|h_0|}\right| \nonumber \\
                   &= |h_0|+|h_1|,
\end{align}
where the last equality holds since $|ab|=|a||b|$ and $\big||a|+|b|\big |=|a|+|b|$.
Next recall that the received signal  from two antennas after post-processing via equal gain receive  combining \cite{Zhang} at the receiver is given by
\begin{eqnarray}
    	y  &=& \frac{|h_0|+|h_1|}{\sqrt{2}} x + z.
\end{eqnarray}
That is, when the unit-power BPSK signal is transmitted, the received signal constellation points becomes

\begin{align}
    	s_0 &=+\frac{|h_0|+|h_1|}{\sqrt{2}}, \ \ \text{and} \ \
s_1  = -\frac{|h_0|+|h_1|}{\sqrt{2}}.
\end{align}
This leads to the distance $d  = \sqrt{2}(|h_0|+|h_1|)$  between the signal points,
which is larger than that of the  B-MBM with unit amplitudes by a factor of $\sqrt{2}$.
Given the channel statistics are identical in both scenarios, we can obtain \eqref{eq:BER_analog_B_MBM} by halving the SNR  in the formula for BER of the BPSK 
in \cite{Zhang}.
\end{proof}

%
%

Next we can state the following result by comparing the BERs in \eqref{eq:BER_B_MBM} and \eqref{eq:BER_analog_B_MBM}, to determine  the SNR gap between the performance of unit-amplitude to that with  the optimal amplitude and phase. The SNR gap is defined as
\begin{align}\label{eq:SNR_gap}
 \textsc{SNR}_\textsc{Gap}={\textsc{snr}_2}/{\textsc{snr}_1},
\end{align}
where  $P_{b, \text{Unit-Amp}}^{\text{(CL B-MBM)}}(\textsc{snr}_2) = P_b^{\text{(CL B-MBM)}}(\textsc{snr}_1)$. That is, the SNR difference
$[\textsc{SNR}_\textsc{Gap} ]_{dB}=[\textsc{snr}_2]_{dB} - [{\textsc{snr}_1}]_{dB}$
indicates the additional power required to sustain the optimal performance using the unit-amplitude weights.
\begin{Cor}\label{cor_gap}
The SNR gap in \eqref{eq:SNR_gap} at high SNR is given by
\begin{align}
 \textsc{SNR}_\textsc{Gap}=\frac{1}{\sqrt{0.75}}\approx   0.6  \ \ \text{dB}.
\end{align}
\end{Cor}

\section{Performance Evaluations}\label{sec:Peformance_simulation}
Table~I summarizes the results of various transmission schemes using BPSK and B-MBM. Fig.~\ref{Fig:BER_uncoded} shows the BER for BPSK and B-MBM over Rayleigh fading channel with unit variance. The analytical results consistently align with the simulations in all cases. The open-loop B-MBM exhibits   a 3 dB gap from the conventional BPSK. However incorporating feedback leads to a significant enhancement in the performance.  For BER of $5\times10^{-3}$, feedback provides nearly 10 dB power gain as compared to the open-loop case. The closed-loop B-MBM with unit-amplitude weights performs very close to the closed-loop B-MBM with optimized complex weights, with loss of just 0.6 dB, which is in agreement with the difference between the high-SNR approximations in Corollary~\ref{cor_gap}. This is encouraging for practical implementations.


\section{Conclusions}\label{sec:concl}


Demonstrating the limitations of the open-loop B-MBM, we analytically illustrated the necessity for feedback to enhance its performance. The transmitter's weights play a crucial role in enlarging the Euclidean distance of the observed constellation symbols at the receiver, enabling improved error protection. The optimal closed-loop B-MBM scheme achieves the performance benchmark established by Alamouti-coded BPSK transmission using two transmit antennas. Notably, even when employing unit-amplitude weights, the closed-loop B-MBM maintains a marginal 0.6 dB deviation from its counterpart with the optimal complex weights. The analytical exploration of higher-order MBM remains a subject for future research.


\begin{figure}
\centering
\vspace{0.10cm}
	\includegraphics[width=.5\textwidth]{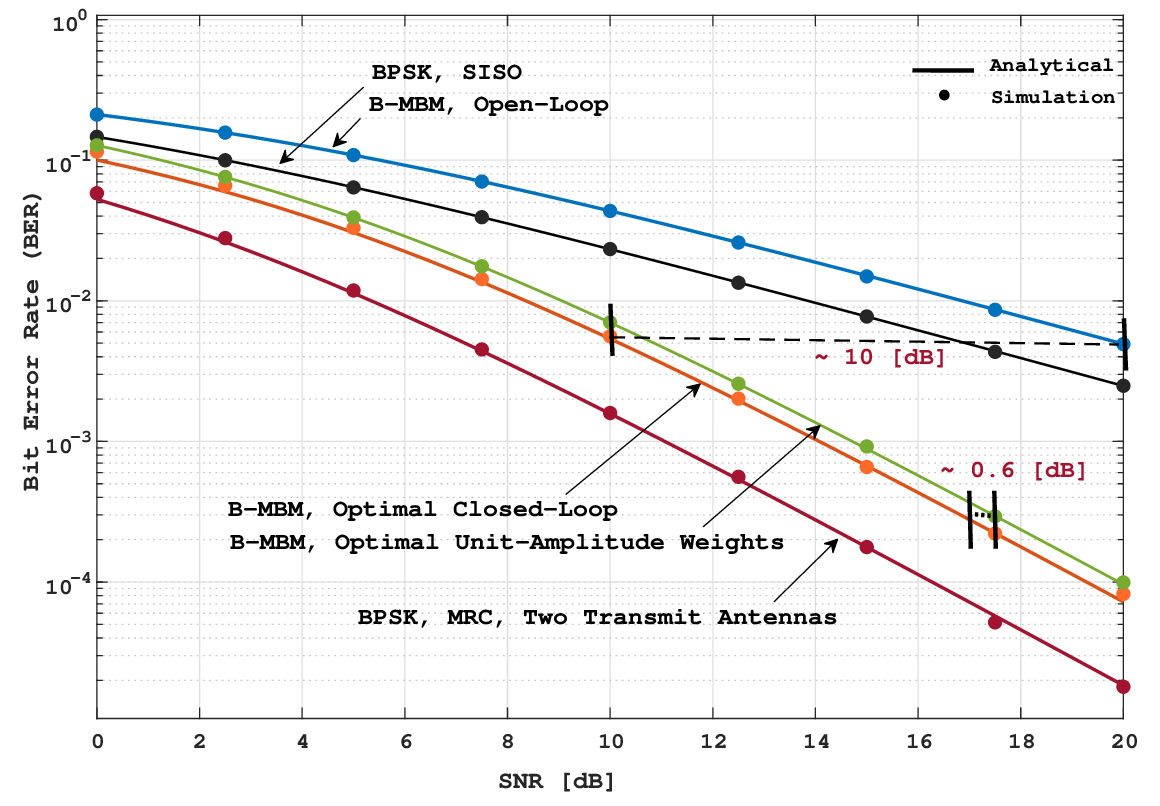}  
\vspace{-0.67cm}
	\caption{Performance of BPSK of SISO (black) and MRC (red) for $2\times 1$ link and B-MBM for open-loop (blue), closed-loop with optimal weights (orange) and closed-loop with optimal unit-amplitude weights (green) with analytical (solid lines) and simulation (dots) results over Rayleigh fading channels.}
	\label{Fig:BER_uncoded}
	\vspace{-0.45cm}
\end{figure}


\end{document}